\documentclass[10pt, draftclsnofoot, onecolumn]{IEEEtran}
\IEEEoverridecommandlockouts
\usepackage{cite}
\usepackage{amsmath,amssymb,amsfonts}
\usepackage{algorithmic}
\usepackage{graphicx}
\usepackage{textcomp}
\usepackage{xcolor}
\usepackage{bbm}
\usepackage{amsthm}
\usepackage{psfrag,caption,subcaption}

\usepackage[left=0.7in, right=0.7in, top=0.7in, bottom=1in]{geometry}
\usepackage{bbm}
\usepackage[ruled,vlined]{algorithm2e}

\usepackage{algorithmic}%

\usepackage[utf8]{inputenc}
\usepackage[english]{babel}
\newtheorem{theorem}{Theorem}

\def\cast{{
   \mathord{
      \hbox to 0em{
         \ooalign{
	   \smash{\hbox{$\ast$}}\crcr
	   \smash{\hskip-1pt\Large\hbox{$\circ$}} }
	 \hidewidth}
      \phantom{\bigcirc}
} }}

\def\bm#1{\mbox{\boldmath $#1$}}

\newcommand{\rH}{^{ \raisebox{1pt}{$\rm \scriptscriptstyle H$}}}
\newcommand{\rT}{^{ \raisebox{1.2pt}{$\rm \scriptstyle T$}}}

\newcommand{\bds}{\begin {itemize}}
\newcommand{\eds}{\end {itemize}}
\newcommand{\bdf}{\begin{definition}}
\newcommand{\blm}{\begin{lemma}}
\newcommand{\edf}{\end{definition}}
\newcommand{\elm}{\end{lemma}}
\newcommand{\bthm}{\begin{theorem}}
\newcommand{\ethm}{\end{theorem}}
\newcommand{\bprp}{\begin{prop}}
\newcommand{\eprp}{\end{prop}}
\newcommand{\bcl}{\begin{claim}}
\newcommand{\ecl}{\end{claim}}
\newcommand{\bcr}{\begin{coro}}
\newcommand{\ecr}{\end{coro}}
\newcommand{\bquest}{\begin{question}}
\newcommand{\equest}{\end{question}}

\newcommand{\larrow}{{\larrow}}

\newcommand{\argmin}{\ensuremath{\mathrm{arg}\min}}
\newcommand{\argmax}{\ensuremath{\mathrm{arg}\max}}

\newcommand{\cC}{{\ensuremath{\mathcal{C}}}}

\newcommand{\cI}{{\ensuremath{\mathcal{I}}}}

\newcommand{\cN}{{\ensuremath{\mathcal{N}}}}

\newcommand{\cP}{{\ensuremath{\mathcal{P}}}}

\newcommand{\cV}{{\ensuremath{\mathcal{V}}}}

\def\mbC{{\ensuremath{\mathbb C}}}

\def\mbE{{\ensuremath{\mathbb E}}}

\newcommand{\vb}{{\ensuremath{{\mathbf{b}}}}}

\newcommand{\vg}{{\ensuremath{{\mathbf{g}}}}}

\newcommand{\vh}{{\ensuremath{{\mathbf{h}}}}}

\newcommand{\vr}{{\ensuremath{{\mathbf{r}}}}}
\newcommand{\vs}{{\ensuremath{{\mathbf{s}}}}}

\newcommand{\vu}{{\ensuremath{{\mathbf{u}}}}}
\newcommand{\vv}{{\ensuremath{{\mathbf{v}}}}}

\newcommand{\vw}{{\ensuremath{{\mathbf{w}}}}}

\newcommand{\vz}{{\ensuremath{{\mathbf{z}}}}}

\newcommand{\mA}{{\ensuremath{\mathbf{A}}}}

\newcommand{\mB}{{\ensuremath{\mathbf{B}}}}

\newcommand{\mC}{{\ensuremath{\mathbf{C}}}}

\newcommand{\mF}{{\ensuremath{\mathbf{F}}}}

\newcommand{\mH}{{\ensuremath{\mathbf{H}}}}

\newcommand{\mI}{{\ensuremath{\mathbf{I}}}}

\newcommand{\mN}{{\ensuremath{\mathbf{N}}}}
\newcommand{\mP}{{\ensuremath{\mathbf{P}}}}

\newcommand{\mR}{{\ensuremath{\mathbf{R}}}}

\newcommand{\mS}{{\ensuremath{\mathbf{S}}}}

\newcommand{\mT}{{\ensuremath{\mathbf{T}}}}

\newcommand{\mW}{{\ensuremath{\mathbf{W}}}}

\newcommand{\mX}{{\ensuremath{\mathbf{X}}}}
\newcommand{\mY}{{\ensuremath{\mathbf{Y}}}}

\usepackage{latexsym}

\def\IC{\mathbb C}
\def\IN{\mathbb N}
\def\IZ{\mathbb Z}
\def\IR{\mathbb R}

\def\shat{^{\mathchoice{}{}%
 {\,\,\smash{\hbox{\lower4pt\hbox{$\widehat{\null}$}}}}%
 {\,\smash{\hbox{\lower3pt\hbox{$\hat{\null}$}}}}}}

\def\bSigma{{
      \ooalign{
      \smash{\hskip.4pt\raise.4pt\hbox{$\Sigma$}}\vphantom{}\crcr
      \smash{\hskip.7pt\raise.6pt\hbox{$\Sigma$}}\vphantom{}\crcr
      \smash{\hbox{$\Sigma$}}\vphantom{$\Sigma$}}
      \vphantom{\hbox{$\Sigma$}}
      }}
\def\bTheta{{
      \ooalign{
      \smash{\hskip.5pt\raise.5pt\hbox{$\Theta$}}\vphantom{}\crcr
      \smash{\hskip.0pt\raise.1pt\hbox{$\Theta$}}\vphantom{}\crcr
      \smash{\hbox{$\Theta$}}\vphantom{$\Theta$}}
      \vphantom{\hbox{$\Theta$}}
      }}
\def\bDelta{{
      \ooalign{
      \smash{\hskip.4pt\raise.4pt\hbox{$\Delta$}}\vphantom{}\crcr
      \smash{\hskip.7pt\raise.6pt\hbox{$\Delta$}}\vphantom{}\crcr
      \smash{\hbox{$\Delta$}}\vphantom{$\Delta$}}
      \vphantom{\hbox{$\Delta$}}
      }}
\def\bLambda{{
      \ooalign{
      \smash{\hskip.5pt\raise.5pt\hbox{$\Lambda$}}\vphantom{}\crcr
      \smash{\hskip.0pt\raise.1pt\hbox{$\Lambda$}}\vphantom{}\crcr
      \smash{\hbox{$\Lambda$}}\vphantom{$\Lambda$}}
      \vphantom{\hbox{$\Lambda$}}
      }}

\makeatletter

\def\bordermatrix#1{\begingroup \m@th
  \@tempdima 8.75\p@
  \setbox\z@\vbox{%
    \def\cr{\crcr\noalign{\kern2\p@\global\let\cr\endline}}%
    \ialign{$##$\hfil\kern2\p@\kern\@tempdima&\thinspace\hfil$##$\hfil
      &&\quad\hfil$##$\hfil\crcr
      \omit\strut\hfil\crcr\noalign{\kern-\baselineskip}%
      #1\crcr\omit\strut\cr}}%
  \setbox\tw@\vbox{\unvcopy\z@\global\setbox\@ne\lastbox}%
  \setbox\tw@\hbox{\unhbox\@ne\unskip\global\setbox\@ne\lastbox}%
  \setbox\tw@\hbox{$\kern\wd\@ne\kern-\@tempdima\left[\kern-\wd\@ne
    \global\setbox\@ne\vbox{\box\@ne\kern2\p@}%
    \vcenter{\kern-\ht\@ne\unvbox\z@\kern-\baselineskip}\,\right]$}%
  \null\;\vbox{\kern\ht\@ne\box\tw@}\endgroup}
\makeatother

\makeatletter
\def\argmin{\mathop{\operator@font arg\,min}}
\def\argmax{\mathop{\operator@font arg\,max}}
\makeatother

\def\bm#1{\mbox{\boldmath $#1$}}

\newcommand{\bbeta}{{\bm\beta}}

\newcommand{\bea}{\begin{array}}
\newcommand{\ena}{\end{array}}
\newcommand{\beq}{\begin{equation}}
\newcommand{\enq}{\end{equation}}

\newcommand{\beqa}{\begin{eqnarray}}
\newcommand{\enqa}{\end{eqnarray}}

\newcommand{\beqan}{\begin{eqnarray*}}
\newcommand{\enqan}{\end{eqnarray*}}

\newcommand{\AL}{\begin{enumerate}}
\newcommand{\ALE}{\end{enumerate}}

\def\addots{\mathinner{
    \mkern1mu\raise0pt\vbox{\kern7pt\hbox{.}}
    \mkern2mu\raise4pt\hbox{.}
    \mkern2mu\raise7pt\hbox{.}
    \mkern1mu}}

\def\sddots{\mathinner{
    \mkern.8mu\raise7pt\hbox{.}
    \mkern.8mu\raise4pt\hbox{.}
    \mkern.8mu\raise0pt\vbox{\kern7pt\hbox{.}}
    \mkern1mu}}

\def\saddots{\mathinner{
    \mkern.2mu\raise0pt\vbox{\kern7pt\hbox{.}}
    \mkern.2mu\raise4pt\hbox{.}
    \mkern.2mu\raise7pt\hbox{.}
    \mkern1mu}}

\def\bvr{\bar{\vr}}
\def\bvb{\bar{\vb}}

\def\sqplus{\mathbin{
	{\ooalign{\hfil\raise.3ex\hbox{\scriptsize
	+}\hfil\crcr\mathhexbox274\crcr\mathhexbox275}}
	}} 
\def\sqminus{\mathbin{
	{\ooalign{\hfil\raise.3ex\hbox{\scriptsize
	--}\hfil\crcr\mathhexbox274\crcr\mathhexbox275}}
	}}

\def\IC{{
   \mathord{
      \hbox to 0em{
	 \hskip-4pt
         \ooalign{
	   \smash{\hskip1.9pt\raise2.6pt\hbox{$\scriptscriptstyle |$}}\crcr
	   \smash{\hbox{\rm\sf C}} }
	 \hidewidth}
      \phantom{\hbox{\rm\sf C}}
} }}
\def\IN{
    {\ooalign{
   \smash{\hskip2.2pt\raise1.5pt\hbox{$\scriptscriptstyle |$}}\vphantom{}\crcr
   \hbox{\sf N}
	}}
	} 
\def\IZ{
    {\ooalign{
   \smash{\hskip1.9pt\raise0pt\hbox{$\sf Z$}}\vphantom{}\crcr
   \hbox{\sf Z}
	}}
	} 
\def\IR{
    {\ooalign{
   \smash{\hskip2.2pt\raise1.5pt\hbox{$\scriptscriptstyle |$}}\vphantom{}\crcr
   \smash{\hskip2.2pt\raise3.3pt\hbox{$\scriptscriptstyle |$}}\vphantom{}\crcr
   \hbox{\sf R}
	}}
	} 

\DeclareMathAlphabet{\mathcmb}{OT1}{cmr}{b}{n}

\def\bSigma{\ensuremath{\mathcmb{\Sigma}}}
\def\bLambda{\ensuremath{\mathcmb{\Lambda}}}

\def\bTheta{\ensuremath{\mathcmb{\Theta}}}

\newcommand{\SI}{\begin{indlist}}
\newcommand{\EI}{\end{indlist}}

\newcommand{\DL}{\begin{dashlist}}
\newcommand{\DLE}{\end{dashlist}}

\makeatletter
\def\setboxz@h{\setbox\z@\hbox}
\def\wdz@{\wd\z@}
\def\boxz@{\box\z@}
\def\underset#1#2{\binrel@{#2}%
  \binrel@@{\mathop{\kern\z@#2}\limits_{#1}}}
\def\binrel@#1{\begingroup
  \setboxz@h{\thinmuskip0mu
    \medmuskip\m@ne mu\thickmuskip\@ne mu
    \setbox\tw@\hbox{$#1\m@th$}\kern-\wd\tw@
    ${}#1{}\m@th$}%
  \edef\@tempa{\endgroup\let\noexpand\binrel@@
    \ifdim\wdz@<\z@ \mathbin
    \else\ifdim\wdz@>\z@ \mathrel
    \else \relax\fi\fi}%
  \@tempa
}
\let\binrel@@\relax%
\makeatother

\newcommand{\Hbr}{\mbox{$\mH_{\rm br}$}}
\newcommand{\Hru}{\mbox{$\mH_{\rm ru}$}}

\newcommand{\Hbu}{\mbox{$\mH_{\rm bu}$}}

\def\BibTeX{{\rm B\kern-.05em{\sc i\kern-.025em b}\kern-.08em
    T\kern-.1667em\lower.7ex\hbox{E}\kern-.125emX}}
\begin{document}

\title{Joint Communication and Radar Sensing with Reconfigurable Intelligent Surfaces
}

\author{\IEEEauthorblockN{ R.S. Prasobh Sankar, Battu Deepak, and  Sundeep Prabhakar Chepuri}
\IEEEauthorblockA{ {\\Indian Institute of Science}, Bangalore, India }

}

\maketitle

\begin{abstract}

In this paper, we use a reconfigurable intelligent surface (RIS) to enhance the radar sensing and communication capabilities of a {\color{black} mmWave} dual function radar communication system. To simultaneously localize the target and to serve the user, we propose to adaptively partition the RIS by reserving separate RIS elements for sensing and communication. We design a multi-stage hierarchical codebook to localize the target while ensuring a strong communication link to the user. We also present a method to choose the number of times to transmit the same beam in each stage to achieve a desired target localization probability of error. The proposed algorithm typically requires fewer transmissions than an exhaustive search scheme to achieve a desired {\color{black} target} localization probability of error. Furthermore, the average spectral efficiency of the user with the proposed algorithm is found to be comparable to that of a RIS-assisted MIMO communication system without sensing capabilities and is much better than that of traditional MIMO systems without RIS.

\end{abstract}


\section{Introduction} \label{sec:intro}

Modern systems envisioned to jointly implement sensing and communication functionalities, also referred to as joint radar communication~(JRC) systems, are gaining significant attention for beyond 5G communications~\cite{mishra2019towards_mmWave,liu2020joint_radar,buzzi2019mimo_jcs}.  The main objective of a JRC system is to enable coexistence between the communication and sensing functionalities by ensuring reliable communications with the users while using the same spectrum for sensing and localizing targets. 

  Usual techniques to enable JRC include the design of beamformers to reduce interference between the radar and communication signals, embedding information symbols in radar waveforms, or using conventional communication symbols to carry out radar sensing~\cite{sodagari2012projection_jrc,liu2020joint_radar,mishra2019towards_mmWave}, to name a few.

A dual function radar communication base station (DFBS) is a type of JRC system that consists of a base station that can simultaneously communicate with the user equipment~(UE) and receive (and process) echo signals reflected from the targets.  A major challenge in operating at the mmWave and higher frequencies is the extreme pathloss. Achieving reasonable radar sensing and communication performance in mmWave DFBS systems thus require techniques to combat such harsh propagation environments. 

A new technology called  reconfigurable intelligent surfaces (RISs) has emerged as a popular choice for mmWave communications\cite{ozdogan2020using,najafi2020physicsbased} as RISs can be used to favourably modify the wireless propagation environment between the transmitter and the receiver.  RIS is a fully passive two-dimensional array consisting of sub-wavelength elements, which can be tuned remotely from the DFBS to introduce certain phase shifts to the electromagnetic waves incident on it and focus the signal energy in desired directions.
 
In this work, we consider a JRC system with a DFBS that communicates to the UE and performs radar sensing with assistance from an RIS. In particular, using the RIS, we sense a point target that is not directly visible to the DFBS and also improve the communication channel between the DFBS and UE.

\subsection{Major contributions} \label{sec:intro:contributions}
To simultaneously localize the target and to reliably communicate with the UE, we propose to  adaptively partition the RIS. During routine surveillance operations, we propose to use a small part of the RIS to transmit wide beams while using the majority of the RIS elements for communications. For target localization, we use a multi-stage hierarchical codebook to transmit beams that progressively get sharper as the stage progresses by reserving a larger portion of the RIS for localization.  We also present a strategy to choose the number of snapshots to be used in each stage to compensate for the varying array gain and to achieve the desired performance in terms of the probability of target localization error. The proposed algorithm typically requires fewer transmissions than an exhaustive search scheme to accomplish a given {\color{black} target} localization probability of error. The average spectral efficiency of the user with the proposed algorithm is found to be comparable to that of a RIS assisted MIMO communication system without sensing capabilities and is much better than that of traditional MIMO systems without RIS.

\section{System model} \label{sec:system}

Consider a setup with one DFBS, which simultaneously serves a user equipment (UE) and  localizes a radar target with assistance from a fully-passive RIS. The DFBS and UE have uniform linear arrays (ULAs) with, respectively, $N_{\rm b}$ and $N_{\rm u}$ elements that are half a wavelength apart.  The RIS is a {\color{black}uniform planar array} (UPA) with $N_{\rm r}$ phase shifters that are spaced less than half a wavelength.

Let us denote the array response vectors of the ULAs at the DFBS and UE towards an angle $\theta$ as $\vb(\theta) \in  \mbC^{N_{\rm b}}$ and $\vu(\theta) \in  \mbC^{N_{\rm u}}$, respectively. The $n$th entry of the steering vectors $\vb(\theta)$ and $\vu(\theta)$ is $e^{j (n-1) \pi \sin \, \theta}$. For an azimuth angle $\phi$ and elevation angle $\psi$, we can define the direction cosines $v_x = \sin\,\psi \sin \, \phi $ and  $v_y = \sin \, \psi \cos\,\phi$, and  collect them to form the direction cosine vector $\vv = [v_x,v_y]\rT$.  The array response vector of the RIS in the direction $\vv$, denoted by $\vr(\vv) \in  \mbC^{N_{\rm r}}$, admits a Kronecker product structure $\vr(\vv) = \vr_x(v_x) \otimes \vr_y(v_y)$ with $\vr_x(v_x) \in \mathbb{C}^{\sqrt{N_{\rm r}}}$ and $\vr_y(v_y) \in \mathbb{C}^{\sqrt{N_{\rm r}}} $ being the array steering vectors of the RIS along the horizontal and vertical directions, respectively. Here, $\otimes$ denotes the Kronecker product. Let us collect {\color{black}the RIS} phase shifts in $\boldsymbol{\omega} \in \mbC^{N_{\rm r}}$. We assume that the RIS phase shifts can also be factorized as  $\boldsymbol{\omega} =  \boldsymbol{\omega}_x \otimes \boldsymbol{\omega}_y$, where each entry of $\boldsymbol{\omega}_x \in \mathbb{C}^{\sqrt{N_{\rm r}}}$ and $\boldsymbol{\omega}_y \in \mathbb{C}^{\sqrt{N_{\rm r}}}$ is unit modulus. {\color{black}  We refer $\boldsymbol{\omega}_x$ and $\boldsymbol{\omega}_y$ as the RIS phase shifts corresponding to the horizontal and vertical directions, respectively.
}

Let us denote the DFBS-UE, DFBS-RIS, and RIS-UE MIMO channels as $\Hbu$, $\Hbr$, and $\Hru$, respectively. These channels are assumed to be known (or estimated~\cite{deepak2021mmWave}). In addition, they are assumed to have only the line-of-sight (LoS) component  {\color{black} due to the severe pathloss at the mmWave frequencies}~\cite{ghosh2014millimeter}.

\subsection{The downlink transmit signal}

The DFBS transmits two data streams collected in $\mS = [\vs_{\rm r},  \vs_{\rm u}]\rT \in \mbC^{2 \times T_s}$ with $T_s$ symbols towards the RIS and UE using a precoding matrix $\mF  = N_{\rm b}^{-1/2} [\vb(\theta_{\rm r}) \,\, \vb(\theta_{\rm u}) ] \in \mbC^{N_{\rm b} \times 2}$, where $\theta_{\rm r}$ and $\theta_{\rm u}$ are the angles of departure of the paths from the DFBS towards the RIS and UE, respectively. The data stream $\vs_{\rm r}$ for the RIS is allocated a power $p_r$ and the data stream $\vs_{\rm u}$ for the UE is allocated a power $p_u$ such that $P = p_r + p_u$. Then the signal transmitted from the DFBS is given by $\mX = \mF\mP\mS = \sqrt{\frac{p_r}{N_{\rm b}}} \vb(\theta_{\rm r}) \vs_{\rm r}\rT+ \sqrt{\frac{p_u}{N_{\rm b}}} \vb(\theta_{\rm u})  \vs_{\rm u}\rT$,
where $\mP = {\rm diag}(\sqrt{p_r} , \sqrt{p_u})$ is the power allocation matrix.

\subsection{The DFBS-RIS-target radar link}

We assume that there is no direct path between the DFBS and target, which we model as a point scatterer. We use the RIS to scan the scattering scene by forming focused beams based on the signals from the DFBS. The signal reflected from the scatterer is focused back to the DFBS via the RIS for processing.

Let $\vv_{\rm b} = [v_{\rm bx}, v_{\rm by}]\rT$ be the direction cosine vector of the path from the DFBS at the RIS. The channel matrix  
$\Hbr$ is given by  
\begin{equation}
\label{eq:hbr}
\Hbr =g_{\rm br}  \vr(\vv_{\rm b}) \vb\rH(\theta_{\rm r}),
\end{equation} 
where  $g_{\rm br} = \beta_{\rm br} \eta_{\rm br}$ is the path attenuation with $\beta_{\rm br}$ and  $\eta_{\rm br}$ being the small-scale and large-scale fadings, respectively. The target response matrix (or the RIS-target channel) for a scatterer in the scan direction $\vv_{\rm t} = [v_{\rm tx}, v_{\rm ty}]\rT$ with respect to the RIS is given by
\begin{equation}
\label{eq:targetresp}
\mT(\vv_{\rm t}) = \gamma\bvr(\vv_{\rm t}) \vr\rH(\vv_{\rm t}), 
\end{equation} 
where $\gamma = \rho \eta_{\rm rt}^2$ is the scattering coefficient {\color{black} with $\eta_{\rm rt}$ and $\rho~\sim~\cC \cN(0,1)$ denoting the large-scale fading and the radar cross section of the scatterer, respectively}. The signal reflected by the scatterer in the direction $\vv_{\rm t}$ and received at the DFBS via the RIS is given by
\begin{equation} \label{eq:radar1}
	\mY_{\rm r} = \mH_{\rm br}\rT\boldsymbol{\Omega}\rT \mT(\vv_{\rm t}) \boldsymbol{\Omega} \mH_{\rm br} \mX + \mN_{\rm r},
\end{equation}
where $\boldsymbol{\Omega} = {\rm diag} (\boldsymbol{\omega}_x \otimes \boldsymbol{\omega}_y)$ and $\mN_{\rm r}$ is the noise matrix at the DFBS whose entries are assumed to be i.i.d. complex Gaussian with zero mean and variance $\sigma_{\rm b}^2$. Using \eqref{eq:hbr} and \eqref{eq:targetresp}, and exploiting the Kronecker structure of the RIS UPA, \eqref{eq:radar1} can be simplified to $\mY_{\rm r} =  \gamma c_x(v_{\rm tx}) c_y(v_{\rm ty})\mB(\theta_{\rm r})\mX + \mN_{\rm r}$,
where the {\it squared spatial responses} of the RIS towards the scan directions $v_{\rm tx}$ and $v_{\rm ty}$ for the DFBS signal impinging on the RIS in the direction $\vv_{\rm b}$ are, respectively, $c_x(v_{\rm tx})= [\vr_x\rH(v_{\rm tx}) {\rm diag}({\boldsymbol \omega}_x ) \vr_x(v_{\rm bx})]^2$ and $c_y(v_{\rm ty})= [\vr_y\rH(v_{\rm ty}) {\rm diag}({\boldsymbol \omega}_y ) \vr_y(v_{\rm by})]^2$, and $\mB(\theta_{\rm r}) = g_{\rm br}^2\bvb(\theta_{\rm r}) \vb\rH(\theta_{\rm r}).$ 

\subsection{The DFBS-UE and DFBS-RIS-UE communication links}

The signal received at the UE, from the DFBS, is given by
 \begin{equation} \label{eq:rec_signal_ue}
 	\mY_{\rm u} = (\mH_{\rm bu} + \mH_{\rm ru}\boldsymbol{\Omega}\mH_{\rm br})\mX + \mN_{\rm u},	
 \end{equation}
where $\mN_{\rm u}$ is the noise matrix at the UE whose entries are assumed to be i.i.d. complex Gaussian with zero mean and variance $\sigma_{\rm u}^2$.
The channel matrices are defined as $\mH_{\rm bu} = g_{\rm bu}\vu(\zeta_{\rm b})\vb\rH(\theta_u) \in \mbC^{N_{\rm u} \times N_{\rm b}}$ and $\mH_{\rm ru} = g_{\rm ru}\vu(\zeta_{\rm r})\vr \rH (\vv_{\rm u}) \in \mbC^{N_{\rm u} \times N_{\rm r}} $, where $\zeta_{\rm r}$ and $\zeta_{\rm b}$ denote the angles of arrival of the paths from the RIS and DFBS at the UE, respectively, and $\vv_{\rm u}=[v_{ux},v_{uy}]\rT$ is the direction cosine vector of the UE at RIS.  The received signal is beamformed using the combiner $\mC =  N_{\rm u}^{-1/2} [\vu(\zeta_{\rm r}) \,\, \vu(\zeta_{\rm b}) ] \in \mbC^{N_{\rm u} \times 2}$.

The main aim of this paper is to localize the target by estimating the direction cosine vector $\vv_{\rm t}$ at the DFBS while ensuring a rank-2 communication {\color{black}channel $\mH_{\rm bu} + \mH_{\rm ru}\boldsymbol{\Omega}\mH_{\rm br}$} between the DFBS and UE.

\section{Codebook design for localization} \label{sec:codebook_and_loc}

In this section, we present the proposed RIS phase shift design  for joint communication and sensing, and present the proposed target localization algorithm. We design the RIS phase shifts to simultaneously scan  the target while ensuring a rank-2 channel towards the UE. To enable this coexistence, we propose to adaptively partition and allocate the RIS elements for communications and sensing by operating in one of the two modes, namely, the {\it surveillance mode} or {\it localization mode}. The proposed partitioning scheme is illustrated in Fig.~\ref{fig1:merge}a. 

In the surveillance mode, to detect the presence of a target, we use widebeams formed using a small number of RIS elements to scan large sectors. The remaining elements are used to serve the UE. If the target is not always present, then we mostly operate in the surveillance mode. In the localization mode, instead of exhaustively searching in different scan directions with fine pencil beams, we  start with a widebeam and progressively make the beams narrower using a larger number of RIS elements to localize the target precisely. Partitioning the RIS elements for sensing and communications results in a fundamental tradeoff between the localization accuracy and the communication spectral efficiency. To perform beam search, we extend the hierarchical codebook based approach from~\cite{alkhateeb2014channel} to perform target localization. 

A disadvantage of allocating fewer RIS elements for localization in the initial stages is the limited array gain, which, nonetheless, can be compensated by performing multiple transmissions with the same beam.

\begin{figure}
\centering
	\includegraphics[scale = 0.35]{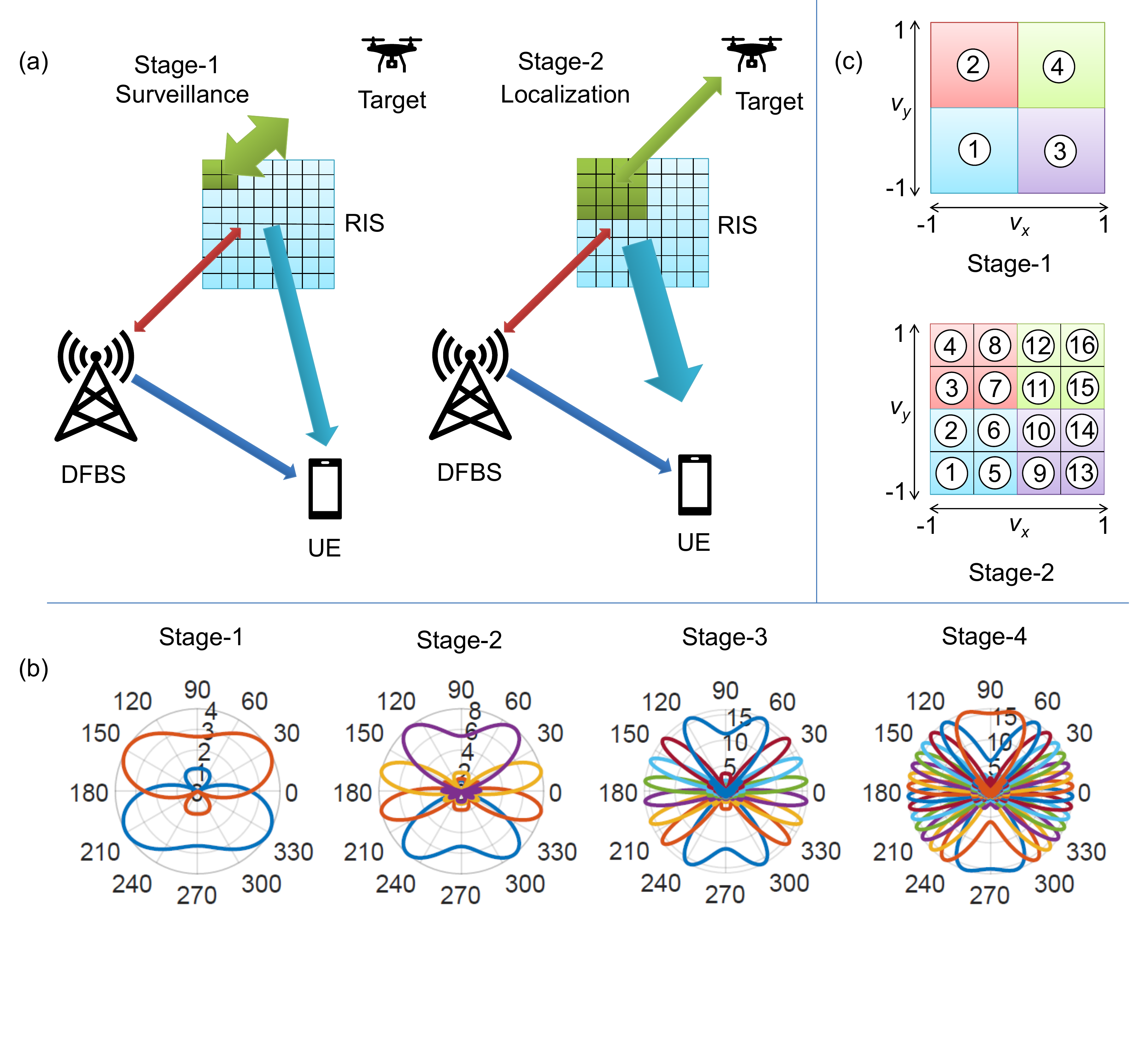}

	\caption{\small (a) Illustration of the partitioning of the RIS for localization and communications. (b) Beam patterns of the first 4 stages of the codebook {\color{black} for localization} (horizontal direction). (c) Partitioning of the $v_x-v_y$ grid for the first two stages. }
	\label{fig1:merge}
	
\end{figure}

\subsection{Codebook design} \label{sec:codebook}

Let us assume that the target direction cosines lie on a discrete grid of $D$ points $\cV = \{v_1, v_2, \ldots, v_D\}$, i.e., $v_x, v_y \in \cV$. The beams for the hierarchical codebook are designed such that a certain portion of the $v_x - v_y$ direction cosine grid is illuminated. We partition the RIS and design separate beams for localization and communication. From the assumed Kronecker structure of the RIS phase shifts $\boldsymbol{\omega} = \boldsymbol{\omega}_x \otimes \boldsymbol{\omega}_y$, we adopt a  simple design procedure to  design  $\boldsymbol{\omega}_x$ and $\boldsymbol{\omega}_y$ independently. 
Next we discuss the procedure to design a unit modulus phase shift $\vw \in \mbC^{\sqrt{N_{\rm r}}}$ that represents either ${\boldsymbol \omega}_x$ or ${\boldsymbol \omega}_y$.

 Let $L_s$ and $C_s$ denote the number of elements of $\vw$ used for localization and communications in the $s$th stage, respectively, such that $L_s + C_s = \sqrt{N_{\rm r}}$. At each stage, we partition $\cV$ into disjoint sets and design the beam to scan a specific angular sector. Particularly, for the $s$th stage, we have $2^s$ partitions with the $i$th partition corresponding to the grid indices $\cI_{s,i}=\left\{  \frac{D}{2^s}(i-1) + 1,  \frac{D}{2^s}(i-1) + 2, \ldots,  \frac{D}{2^s}i   \right\}$ for $i=1,2,\ldots, 2^s$. Let $\vw_{s,i} = \begin{bmatrix}
	\vg_{s,i}\rT & \vh_{s,i}\rT
\end{bmatrix}\rT  \in \mbC^{\sqrt{N_{\rm r}}}$ denote the $i$th beam in the $s$th stage with $\vg_{s,i} \in \mbC^{L_s}$ and $\vh_{s,i} \in \mbC^{C_s}$, being the part of the beamformer corresponding to the radar and communications, respectively. To design the beamformer to scan a certain partition of the direction cosine space, we choose $\vg_{s,i}$ such that
\begin{equation} \label{eq:w_design_0}
     \vr_x\rH(v_j){\rm diag}(\vg_{s,i})\vr_x(v_{bx}) = \begin{cases}
     	L_s \quad &\text{for} \quad v_j \in \cI_{s,i} \\ 0 \quad &\text{otherwise}
     \end{cases},
\end{equation}
where $j=1,2,\ldots,D$. Defining the RIS array response corresponding to the directions in $\cV$ as $\mR = [\vr_x(v_1), \ldots , \vr_x(v_D)] \in \mbC^{L_s \times D}$, ~\eqref{eq:w_design_0} can be written as
\begin{equation} \label{eq:w_design_1}
	\mR\rH {\rm diag}(\vg_{s,i})\vr_x(v_{bx}) = L_s{\mathbbm{1} }_{s,i},
\end{equation} 
where $\mathbbm{1}_{s,i} \in \{1,0\}^{D}$ is an indicator vector containing ones at the indices corresponding to $\cI_{s,i}$ and zeros elsewhere. Let $\circ$ denote the Khatri-Rao product. Then, using the property ${\rm vec}(\mA {\rm diag}(\vg) \mB) = (\mB\rT \circ \mA) \vg$, we have $(\vr_x\rT(v_{bx}) \circ \mR\rH) \vg_{s,i} =   L_s {\mathbbm{1} }_{s,i}.$

Now, the design of $\vg_{s,i}$ can be stated as the solution to the following optimization problem
\begin{align} 
	&\underset{\vg_{s,i} \in \mathbb{C}^{L_s}  }{\rm minimize} \quad  \Vert (\vr_x\rT(v_{bx}) \circ \mR\rH) \vg_{s,i} -  L_s \mathbbm{1}_{s,i}  \Vert^2  \nonumber  \\
	& \text{subject to } \quad  \vert [\vg_{s,i}]_m \vert = 1,~m=1,2,\ldots,L_s,
	\label{eq:des_w_x}
\end{align}
where we have unit modulus constraints because the RIS has only  passive elements. The above problem can be iteratively solved using updates of the form~\cite[Algorithm 1]{tranter2017fast_uls}
\[
\vg_{s,i}^{(k+1)} = {\rm exp} \left( j\, {\rm angle} \left\{ \vg_{s,i}^{(k)} + \mu \mA\rH ( L_s \mathbbm{1}_{s,i} - \mA\vg_{s,i}^{(k)} )    \right\} \right),
\]
where $\mA = (\vr_x\rT(v_{bx}) \circ \mR\rH)$ and $\mu$ is the step size. The beampatterns for first few stages are shown in   Fig. \ref{fig1:merge}b, where we use $(L_1,L_2,L_3,L_4) = (4,8,16,16)$. We can observe that the beams become narrower as the stage index increases. By reserving only a small number of RIS elements, we obtain sufficiently wide beam{\color{black}s} in the initial stages, albeit a smaller array gain. 

The design of the phase shifts for the communication part of the RIS, i.e., $\vh_{s,i}$, to form a pencil beam towards the UE to ensure a rank-2 communication channel  $\mH_{\rm bu} + \mH_{\rm ru}{\rm diag}(\vw_{s,i})\mH_{\rm br}$, is relatively straightforward.
Using the Kronecker factorization of $\vw$, the gain offered by the RIS elements reserved for communication to the UE in the horizontal direction can be computed as ${\color{black}c_{ux}} = {\vr}_{xc}\rH(v_{ux}) {\rm diag}(\vh_{s,i}) {\vr}_{xc}(v_{bx})  $, where ${\vr}_{xc}(.) \in \mbC^{C_s}$ denotes the last $C_s$ elements of the RIS array response vector $\vr_x(.) \in \mbC^{\sqrt{N_{\rm r}}}$. To focus the entire signal energy impinging on the RIS from direction $v_{bx}$ to the direction $v_{ux}$, we choose the phase shifts such that  {\color{black}$c_{ux}$} is maximized. It can be shown that { \color{black} $c_{ux}$ } attains its maximum value $C_s$ when $\vh_{s,i} = \bar{ \vr}_{xc}(v_{bx}) \odot \vr_{xc}(v_{ux})$, where $\odot$ is the Hadamard product. 

Let us denote the possible beams in the $s$th stage  as $\mW_{s} = \begin{bmatrix}
	\vw_{s,1} & \vw_{s,2} & \ldots & \vw_{s,2^s}
\end{bmatrix} \in \mbC^{\sqrt{N_{\rm r}} \times 2^s}  $.
Then the desired set of two-dimensional beams in the $s$th stage, which illuminates a certain partition of the $v_x - v_y$ grid while forming a pencil beam to the UE  is given by $	\boldsymbol{\Omega}_s = \begin{bmatrix}
	\boldsymbol{\omega}_{s,1} & \ldots & \boldsymbol{\omega}_{s,4^s} \end{bmatrix} = \mW_{s} \otimes \mW_{s} \in \mbC^{N_{\rm r} \times 4^s}$,
where $\boldsymbol{\omega}_{s,i} \in \mbC^{N_{\rm r}}$ denotes the RIS phase shifts corresponding to the $i$th beam in the $s$th stage. Since each beam $\vw_{s,i}$ essentially scans a certain partition in the $v_x$ (or $v_y$) space, each column of $\mW_s$ scans a certain partition of the $v_x - v_y$ grid. 
For example, since $\vw_{1,1}$ covers the range $-1 \leq v_x \leq 0$ (or $-1 \leq v_y \leq 0$), $\boldsymbol{\omega}_{(1,1)} = \vw_{1,1} \otimes \vw_{1,1}$ covers the range $-1 \leq v_x,v_y \leq 0$, as illustrated in Fig. \ref{fig1:merge}c. Using the designed beams, we now discuss the target localization algorithm.

\subsection{Target localization} \label{eq:target_loc}
In the first stage, we transmit four beams corresponding to the four columns of $\boldsymbol{\Omega}_1$ and compare the power of the  received signal at the DFBS. The partition of the $v_x-v_y$ space corresponding to the largest received power is selected and that is further divided into four more parts. In the second stage, we then transmit four beams corresponding to this partition. For example, if the 3rd beam in the first stage resulted in the largest received power, we will be selecting beams $\{9,10,13,14\}$ for the second stage as illustrated in Fig. \ref{fig1:merge}c. This process is then repeated for other stages.

In general, let $\cP_s = \{p_{s1}, p_{s2}, p_{s3}, p_{s4}\}$ be the indices of the beams to be transmitted during the $s$th stage.
 Then the signal received at the DFBS related to the $i$th beam is given by $	\mY_{\rm r}^{(s,i)} = \mH_{\rm br}\rT{\rm diag}(\boldsymbol{\omega}_{s,{p_{si}}})\rT \mT(\vv_{\rm t}) {\rm diag}(\boldsymbol{\omega}_{s,{p_{si}}}) \mH_{\rm br} \mX + \mN_{\rm r}$ for $i=1,2,3,4$.
This signal is then beamformed using $\bvb(\theta_r)$ and correlated with the transmitted signal $\vs_{\rm r}\rT$ to obtain $\vz_{s,i} = \left(\vb\rT(\theta_r)\mY_{\rm r}^{(s,i)} {\rm diag}( {\bar{\vs}_{\rm r}}) \right)\rT \in \mbC^{T_s},~i=1,2,3,4$. We now identify the partition of the $v_x-v_y$ space with the largest received power by computing the average power received for each beam as $\tau_s = \underset{i}{ \text{arg max}  }\quad \left\vert \frac{1}{T_s}\sum_{m=1}^{T_s}   [\vz_{s,i}]_m \right\vert^2$.
Based on the value of $\tau_s$, we then choose the set of indices for the   $(s+1)$th stage, $\cP_{s+1}$, and so on. We finally estimate the target location $\hat{\vv}_t$ from the partition chosen at the last stage $N_{\rm s} = {\rm log_2} N_{\rm d} $.

 Due to the limited array gain for localization in the initial stages of the codebook, an appropriate selection of $T_s$ is crucial for the performance of the proposed algorithm. Next, we  present as a theorem a way to select the number of snapshots $T_s$ in each stage to achieve a desired {\color{black} target localization} probability of error.

\subsection{Design for the desired probability of error} \label{sec:performance}
  Let $(i,j)$ and $(\hat{i},\hat{j})$ denote the position of the partition in the $v_x-v_y$ grid corresponding to the true direction cosine $\vv_{\rm t}$ and the estimated direction cosine $\hat{\vv}_{\rm t}$, respectively. The  target localization probability of error  is defined as $	E = {\rm Pr}( (\hat{i},\hat{j})  \neq (i,j) )$. Similarly, let $(i_s,j_s)$ and $(\hat{i}_s,\hat{j}_s)$ represent the true and estimated partitions of the $v_x-v_y$ grid in the $s$th stage of the proposed algorithm. Then the probability of error in the $s$th stage is defined as $E_s = {\rm Pr} ((\hat{i}_s,\hat{j}_s) \neq ( i_s, j_s ) )$.
  
\begin{theorem}
	Let us assume that the horizontal and the vertical beamformers satisfy~\eqref{eq:w_design_1}. Then the probability of error at each stage, $E_s$, is less than or equal to $\delta$ if the  power allocated towards the RIS, $p_r$, and the number of transmissions  $T_s$ satisfies
	\begin{equation} \label{eq:power_alloc}
		p_rT_s = \left(\frac{\kappa}{1-\kappa}\right)\frac{\sigma_{b}^2}{N_{\rm b}^2 L_s^8 \eta_{\rm br}^2 \eta_{\rm rt}^2  },
	\end{equation}
	where $\kappa = 2(1-\frac{2}{3}\delta)$.
\end{theorem}
\begin{proof}
	Using the results from~\cite[equations (31)-(40), Theorem 1]{alkhateeb2014channel} and~\cite[equations (22), (24)-(35)]{simon1998unified}, we can obtain an upper bound on the average error probability for the $s$th stage, $E_s$. 
	When the horizontal and vertical codebook  satisfies~\eqref{eq:w_design_1}, the beams designed in each stage partition the $v_x-v_y$ plane into $4^{s}$ non-overlapping parts as shown in Fig. \ref{fig1:merge}c. This allows us to obtain a simplified expression for $E_s$ similar to the one obtained in~\cite[Corollary 2]{alkhateeb2014channel}.  To ensure that the probability of error in each stage is bounded, we can set the upper bound obtained to $\delta$. 
\end{proof}
Since the probability of error in each stage is upper bounded by $\delta$, we can show that the overall probability of error is upper bounded by $N_{\rm s}\delta$ using the union bound. Hence, given $p_r$, Theorem 1 allows us to  select $T_s$ to achieve a desired target localization performance.

\section{Numerical experiments} \label{sec:numerical_experiments}

\begin{table}[t]
	\begin{center}
		\small
		\begin{tabular}{|l | l|} 
			\hline
			($N_{\rm b}, N_{\rm u}, N_{\rm r}$, $D$) & ($64$, $16$, $64 \times 64, 32$) \\ \hline
			($\theta_{\rm r}, \theta_{\rm u}, \zeta_{\rm b}, \zeta_{\rm r}$)
		& ($45^0, -25^0, -30^0, 25^0   $) \\
		\hline
		($v_{ux}, v_{uy}, v_{bx}, v_{by}$) & ($ 0.105 , -0.343, 0.133, -0.112 $)	 \\
		\hline
		($d_{\rm bu}, d_{\rm br}, d_{\rm ru}, d_{\rm rt}$) & ($20~{\rm m},10~{\rm m},10~{\rm m},5~{\rm m}$) \\ \hline
		($\alpha_{\rm bu},\alpha_{\rm ru},\alpha_{\rm rt}, \alpha_{\rm br}$) & ($3.5,2.8,2.8,2.5$) \\ \hline
		($\sigma_{b}^2, \sigma_{u}^2$) & ($-94~{\rm dBm}, -80~{\rm dBm}$) \\ \hline	 		
				\end{tabular}
	\end{center}

	\caption{\small  Parameters used for the simulations.}
	\label{table:2}

\end{table}



\begin{figure*}[t]
\begin{subfigure}[c]{0.5\columnwidth}\centering
   \includegraphics[width=\columnwidth]{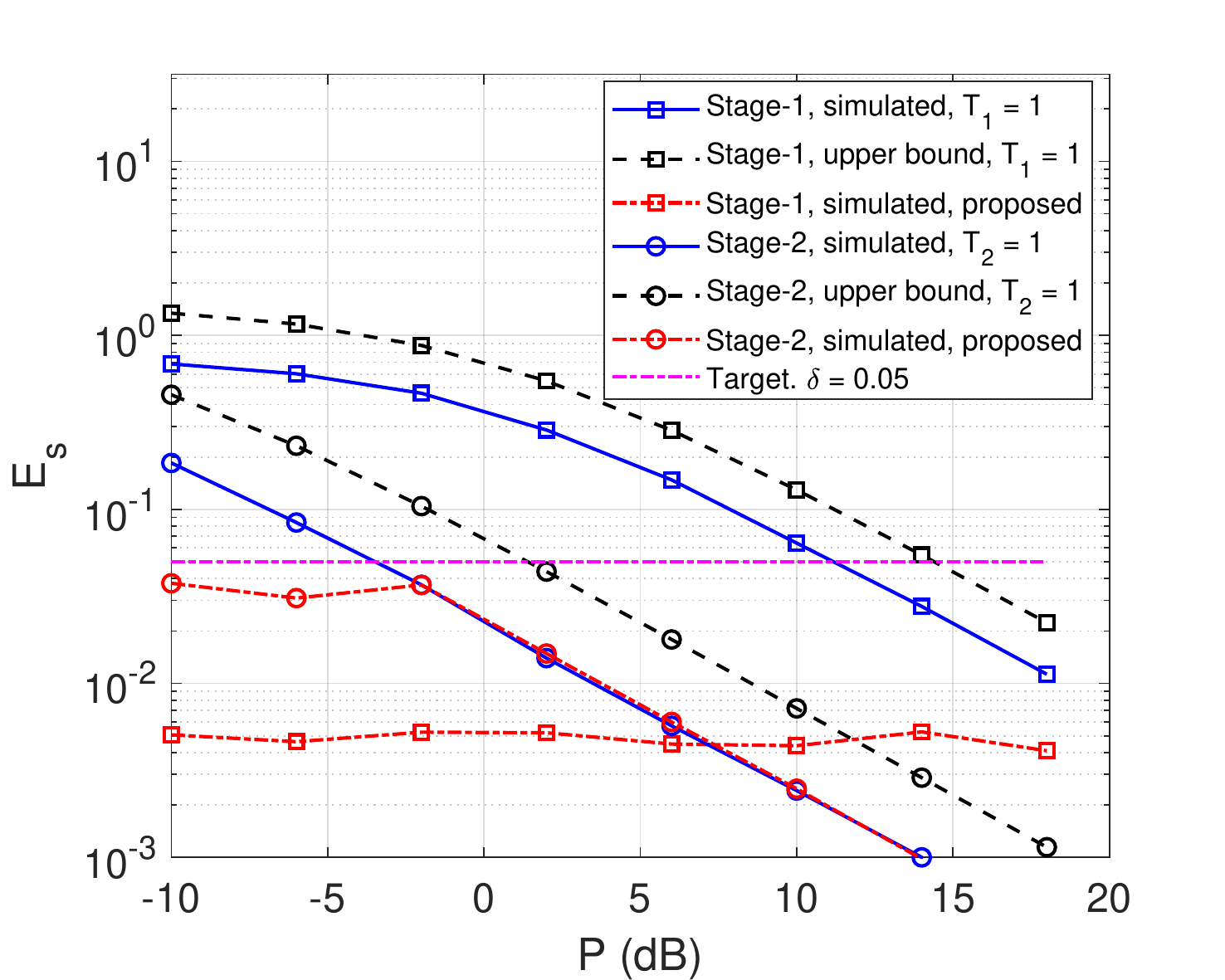}
   \caption{}
   \label{fig2a}
\end{subfigure}
~
\begin{subfigure}[c]{0.5\columnwidth} \centering
   \includegraphics[width=\columnwidth]{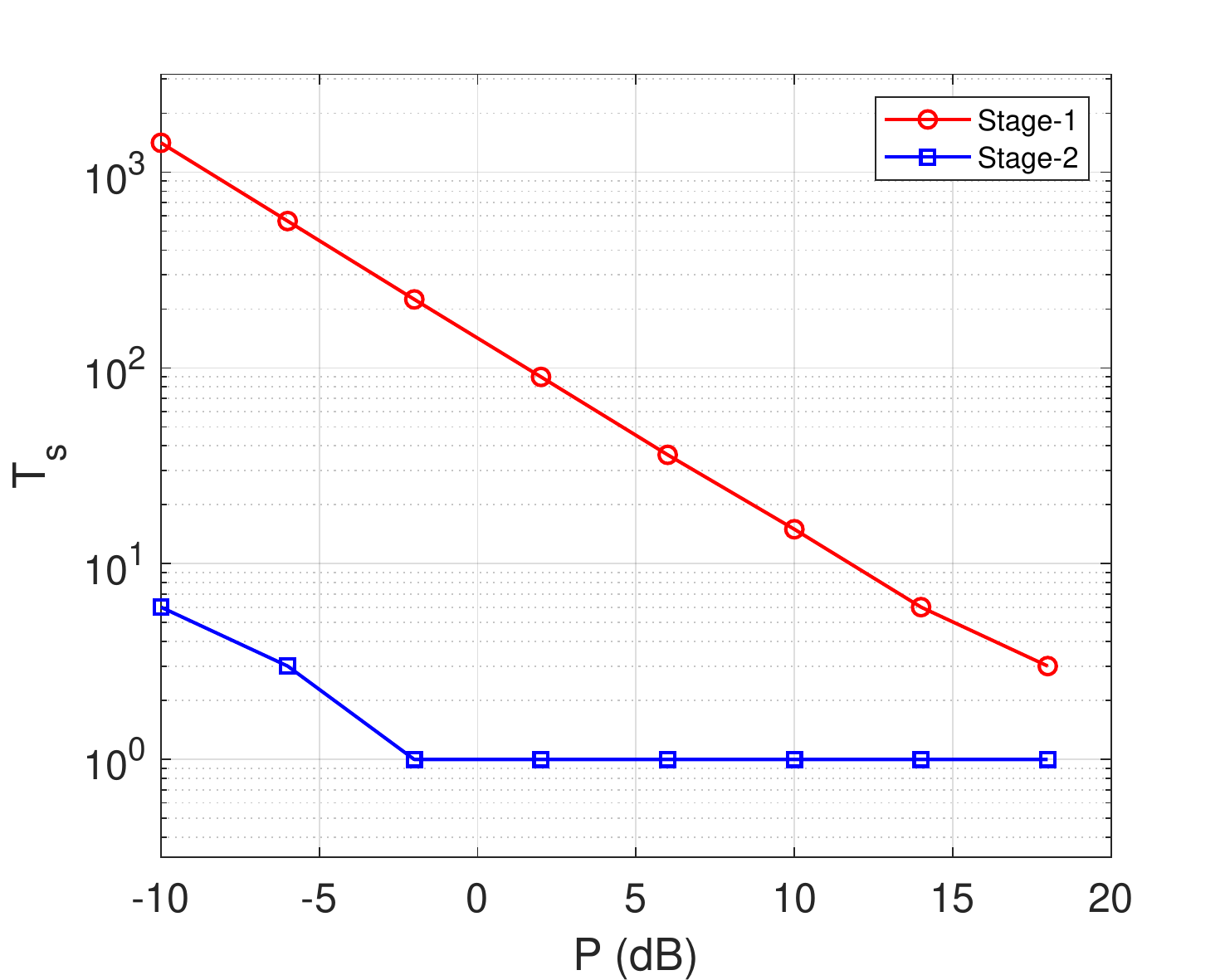}
   \caption{}
   \label{fig2b} 
\end{subfigure}
~
\begin{subfigure}[c]{0.5\columnwidth} \centering
   \includegraphics[width=\columnwidth]{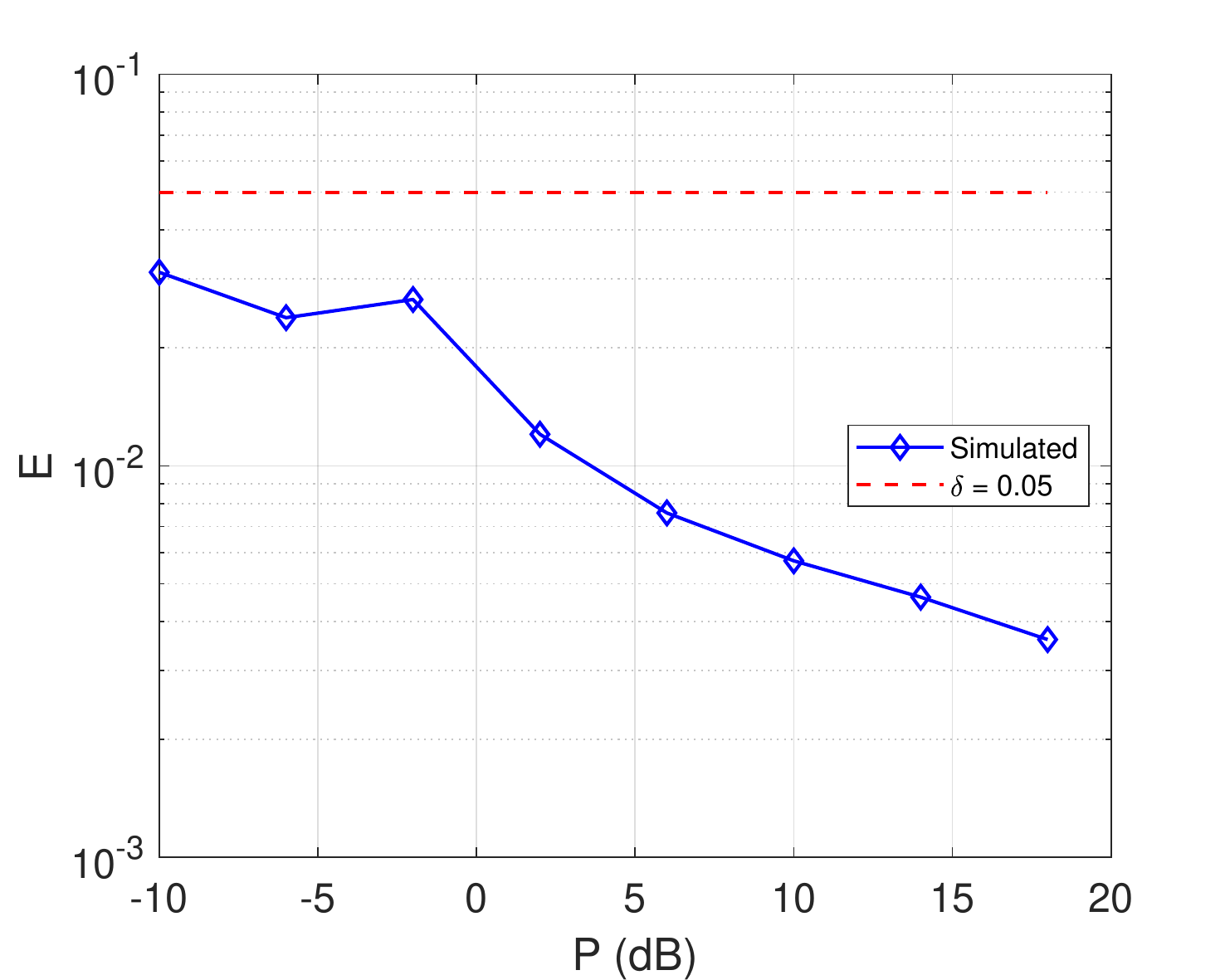}
   \caption{}
   \label{fig2c}
\end{subfigure}
~
\begin{subfigure}[c]{0.5\columnwidth} \centering
   \includegraphics[width=\columnwidth]{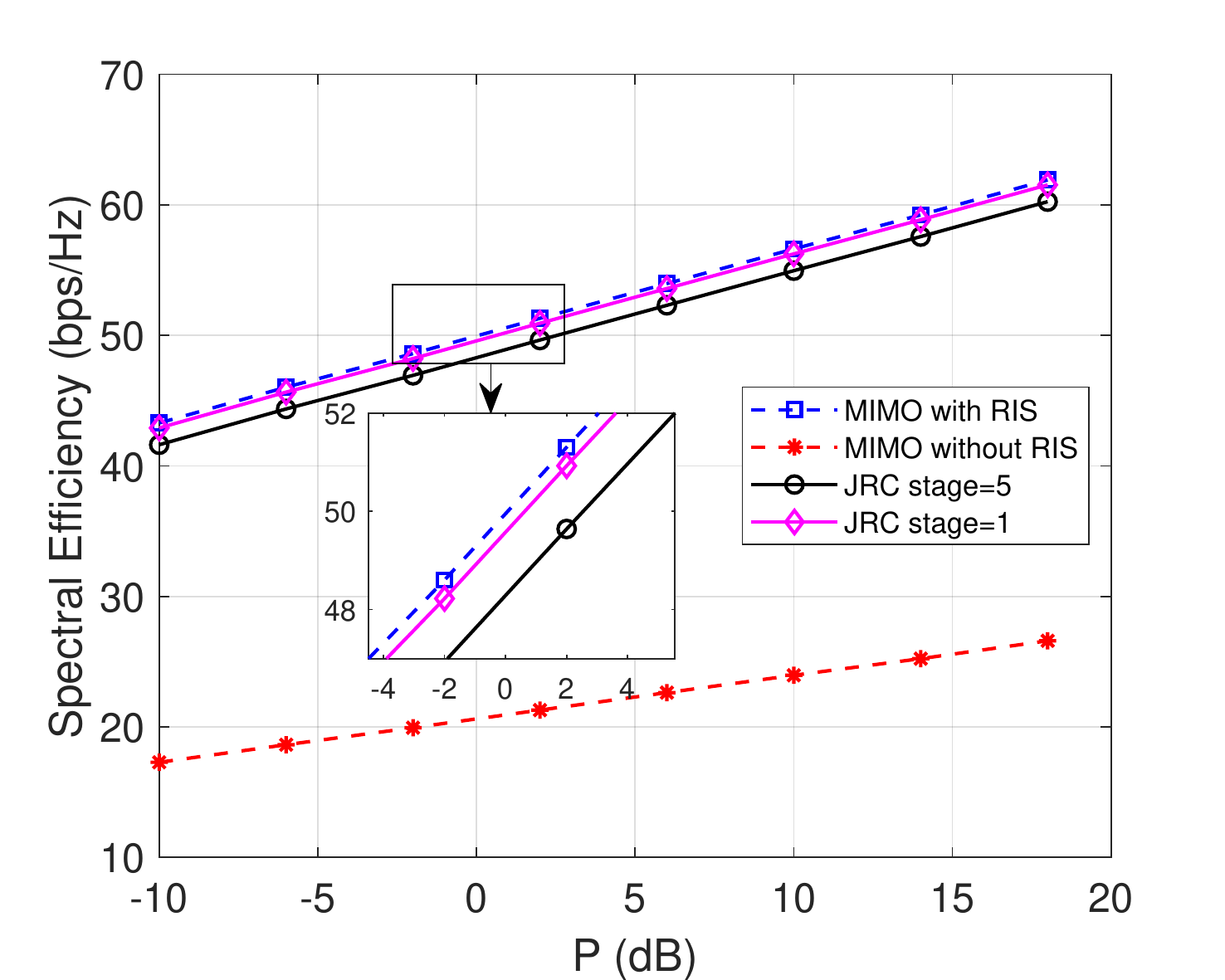}
   \caption{}
   \label{fig2d}
\end{subfigure}
\caption{ (a) Stagewise probability of error. (b) Required number of snapshots to achieve $\delta = 0.05$. (c) Probability of error. (d) Spectral efficiency. }

\end{figure*}

In this section, we illustrate the performance of the proposed target localization algorithm through numerical simulations. An inter{\color{black}-}element spacing of $\lambda/2$ and $\lambda/4$ are used for the ULA and RIS, respectively, where $\lambda$ is the signal wavelength. The search grid $\cV$ is obtained by choosing $D=32$ equally spaced points from $[-1,1]$. We choose $v_{tx}  = -0.4127$, and $v_{ty}  = 0.5397$, with corresponding angles being $\psi_{\rm rt} = 42.79^{0}$ and $\phi_{\rm rt} = -37.40^0$.    Pathloss for each path is modeled as $\eta_{\rm xy} = (\frac{\eta_{0}}{d_{xy}})^{\alpha_{xy}}$, where $\eta_{0} = -30~{\rm dB}$ is the reference path loss at $1~{\rm m}$, $\alpha_{xy}$ is the path loss exponent and $d_{xy}$ is the distance between the considered terminals in metres (${\rm m}$) with $x,y\in \{\rm b,r,u\}$. The number of elements reserved for localization in each stage is selected as 	$\left(L_1, L_2, L_3, L_4, L_5\right) = \left(4,8,16,16,16\right)$.
Details of the parameters used for the simulation are summarized in Table \ref{table:2}.   Since the SNR at the DFBS and UE vary with different stages of operation, we present the simulations results as a function of  $P=p_r + p_u$ with $\mP = 0.5\mI$, i.e., $p_r = p_u = 0.5P$.

In Fig. \ref{fig2a}, we present the probability of error for the first two stages of the proposed algorithm when we use $1$ snapshot for all stages (i.e., $T_s = 1,~\forall s$, indicated as \texttt{simulated, $T_s$ = 1})  and when we choose $T_s$ based on \eqref{eq:power_alloc} (indicated as \texttt{simulated, proposed}) to achieve a desired stagewise error probability of $\delta = 0.05$ for a specific $P$.

We can observe that it is not possible to achieve the desired error probability with a single snapshot, especially for the first stage due to the limited array gain for localization.  Fig. \ref{fig2b} shows the  value of $T_s$ for the first two stages computed using \eqref{eq:power_alloc}. As expected, the second stage requires fewer snapshots due to the increased array gain. Specifically, for the considered range of $P$, we found that $T_3=T_4=T_5=1$ would already achieve the desired error probability, and the total number of snapshots $\sum_{k=1}^{5}T_k$ is comparable to $T_1$.

Probability of error for the proposed algorithm after $N_{\rm s}$ stages is presented in Fig. \ref{fig2c}. We can observe that the proposed power-snapshot allocation presented in~\eqref{eq:power_alloc} ensures that the probability of error is  within the tolerance. For reasonable levels of target detection SNR, the number of snapshots required to perform localization using the proposed codebook is much less than what is required by the exhaustive search. For example, with $P = 6~{\rm dB}$ and $L_1 = 4$ and for the values of the receiver noise and path loss given in Table~\ref{table:2}, we get an approximate target detection SNR of $2~{\rm dB}$ at the DFBS. Total number of transmissions required by the proposed codebook is $(36+1+1+1+1)\times 4 = 160$. However, an exhaustive search scheme, where the space is scanned using pencil beams corresponding to each of the $D=32$ directions in the horizontal and vertical directions, requires $32\times 32 = 1024$ transmissions.

Impact of the localization operation on communications with the UE is characterized by computing the average spectral efficiency~(SE) of the UE, which is defined as $	{\rm SE} = \mbE\left[{\rm log_2} \,\, {\rm det} \left( \mI_{2} +  \frac{1}{\sigma_{u}^2} \mH_{\rm eff} \mH_{\rm eff}\rH     \right)  \right]$,
where $\mH_{\rm eff} = \mC\rH (\mH_{\rm bu} + \mH_{\rm ru}\boldsymbol{\Omega}\mH_{\rm br})\mF\mP$, and the expectation is carried out over different realizations of the small-scale fading coefficients, which are assumed to be Gaussian distributed with  unit variance.
The best SE is obtained when the least number of RIS elements are allocated for  localization.

 The SE of the UE for the proposed JRC system is presented in Fig. \ref{fig2d}. We can observe that the best case average SE obtained during Stage-1 is comparable to the benchmark scheme, which is  a conventional RIS-assisted MIMO communication system without radar sensing capabilities. This also means that the SE achievable by the UE is on par with that of a system without sensing capabilities when the DFBS is operating in the surveillance mode.  Similarly, we can observe that the worst case SE obtained in Stage-5 is about $3$ dB less than what is achieved by the benchmark scheme and is much better than the SE of a MIMO system without RIS.

\section{Conclusions} \label{sec:conclusions}
In this paper, we used  RIS to enable the coexistence between communication and radar sensing in a mmWave DFRC MIMO system. To simultaneously serve the UE and to localize the target, we dynamically partitioned the RIS into two parts for localization and communication. We developed a two-dimensional hierarchical codebook for target localization and presented a method to choose the number of snapshots to be used in each stage to achieve a desired target localization probability of error. For reasonable SNRs and target localization errors, the proposed algorithm requires fewer transmissions than an exhaustive search scheme. The average SE of the UE for the proposed algorithm is found to be comparable to that of a RIS-assisted MIMO communication system without sensing capabilities and is much better than that of traditional MIMO systems without RIS.

\bibliographystyle{ieeetran}
\bibliography{IEEEabrv,bibliography}

\end{document}